\documentclass[10pt, conference]{IEEEtran}
\usepackage{subfigure,stfloats}
\usepackage{epsfig,graphicx,psfrag}
\usepackage{amsfonts,amsmath,amssymb}

\newtheorem{theorem}{Theorem}[section]
\newtheorem{lemma}[theorem]{Lemma}

\begin{document}

\title{Cascade Multiterminal Source Coding}

\author{
\authorblockN{Paul Cuff, Han-I Su, and Abbas El Gamal}
\authorblockA{Department of Electrical Engineering\\
Stanford University\\
E-mail: \{cuff, hanisu, abbas\}@stanford.edu }
}

\maketitle

\begin{abstract}
We investigate distributed source coding of two correlated sources $X$ and $Y$ where messages are passed to a decoder in a cascade fashion.  The encoder of $X$ sends a message at rate $R_1$ to the encoder of $Y$.  The encoder of $Y$ then sends a message to the decoder at rate $R_2$ based both on $Y$ and on the message it received about $X$.  The decoder's task is to estimate a function of $X$ and $Y$.  For example, we consider the minimum mean squared-error distortion when encoding the sum of jointly Gaussian random variables under these constraints.  We also characterize the rates needed to reconstruct a function of $X$ and $Y$ losslessly.

Our general contribution toward understanding the limits of the cascade multiterminal source coding network is in the form of inner and outer bounds on the achievable rate region for satisfying a distortion constraint for an arbitrary distortion function $d(x,y,z)$.  The inner bound makes use of a balance between two encoding tactics---relaying the information about $X$ and recompressing the information about $X$ jointly with $Y$.  In the Gaussian case, a threshold is discovered for identifying which of the two extreme strategies optimizes the inner bound.  Relaying outperforms recompressing the sum at the relay for some rate pairs if the variance of $X$ is greater than the variance of $Y$.
\end{abstract}

\section{Introduction}
\label{section introduction}

Distributed data collection, such as aggregating measurements in a sensor network, has been investigated from many angles \cite{Giridhar--Kumar2006}.  Various algorithms exist for passing messages to neighbors in order to collect information or compute functions of data.  Here we join in the investigation of the minimum descriptions needed to quantize and collect data in a network, and we do so by studying a particular small network.  These results provide insight for optimal communication strategies in larger networks.

In the network considered here, two sources of information are to be described by separate encoders and passed to a single decoder in a cascade fashion.  That is, after receiving a message from the first encoder, the second encoder creates a final message that summarizes the information available about both sources and sends it to the decoder.  We refer to this setup as the {\em cascade multiterminal source coding} network, shown in Figure \ref{figure cascade multiterminal}.  Discrete i.i.d. sources $X_i \in {\cal X}$ and $Y_i \in {\cal Y}$ are jointly distributed according to the probability mass function $p_0(x,y)$.  Encoder 1 summarizes a block of $n$ symbols $X^n$ with a message $I \in \{1,...,2^{nR_1}\}$ and sends it to Encoder 2.  After receiving the message, Encoder 2 sends an index $J \in \{1,...,2^{nR_2}\}$ to describe what it knows about both sources to the decoder, based on the message $I$ and on the observations $Y^n$.  The decoder then uses the index $J$ to construct a sequence $Z^n$, where each $Z_i$ is an estimate of a desired function of $X_i$ and $Y_i$.

\begin{figure}
\psfrag{l1}[][][0.7]{Encoder 1}
\psfrag{l2}[][][0.7]{Encoder 2}
\psfrag{l3}[][][0.7]{Decoder}
\psfrag{l4}[][][0.9]{$X^n$}
\psfrag{l5}[][][0.9]{$Y^n$}
\psfrag{l6}[][][0.9]{$Z^n$}
\psfrag{l7}[][][0.8]{$i(X^n)$}
\psfrag{l8}[][][0.8]{$j(I,Y^n)$}
\psfrag{l9}[][][0.8]{$z^n(J)$}
\psfrag{l11}[][][0.8]{$I \in [2^{nR_1}]$}
\psfrag{l12}[][][0.8]{$J \in [2^{nR_2}]$}
\centering
\includegraphics[width=3.0in]{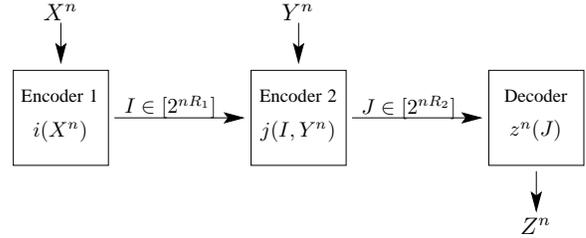}
\caption{{\em Cascade Multiterminal Source Coding.}  The i.i.d. source sequences $X_1,...,X_n$ and $Y_1,...,Y_n$ are jointly distributed according to $p_0(x,y)$.  Encoder 1 sends a message $I$ about the sequence $X_1,...,X_n$ at rate $R_1$ to Encoder 2.  The second encoder then sends a message $J$ about both source sequences at rate $R_2$ to the decoder.  We investigate the rates required to produce a sequence $Z_1,...,Z_n$ with various goals in mind, such as reconstructing estimates of $X^n$ or $Y^n$ or a function of the two.}
\label{figure cascade multiterminal}
\end{figure}

For example, consider the lossless case.  Suppose we wish to compute a function of $X$ and $Y$ in the cascade multiterminal source coding network.  What rates are needed to reliably calculate $Z_i = f(X_i,Y_i)$ at the decoder?  Computing functions of observations in a network has been considered in various other settings, such as the two-node back-and-forth setting of \cite{Ma--Ishwar2008} and the multiple access channel setting in \cite{Nazer--Gastpar2007}.  In the cascade multiterminal network, the answer breaks down quite intuitively.  For the message from Encoder 1 to Encoder 2, use Wyner-Ziv encoding \cite{WynerZiv76} to communicate the function values.  Then apply lossless compression to the function values at Encoder 2.  Computing functions of data in a Wyner-Ziv setting was introduced by Yamamoto \cite{Yamamoto82}, and the optimal rate for lossless computation was shown by Orlitsky and Roche \cite{Orlitsky01} to be the conditional graph entropy on an appropriate graph.

A particular function for which the optimal rates are easy to identify is the encoding of binary sums of binary symmetric $X$ and $Y$ that are equal with probability $p$, as proposed by Korner and Marton \cite{Korner79}.  For this computation, the required rates are $R_1 \geq h(p)$ and $R_2 \geq h(p)$, where $h$ is the binary entropy function.  Curiously, the same rates are required in the standard multiterminal source coding setting.

Encoding of information sources at separate encoders has attracted a lot of attention in the information theory community over the years.  The results of Slepian-Wolf encoding and communication through the Multiple Access Channel (MAC) are surprising and encouraging.  Slepian and Wolf \cite{Slepian_wolf_73_source_coding} showed that separate encoders can compress correlated sources losslessly at the same rate as a single encoder.  Ahlswede \cite{Ahlswede73MAC} and Liao \cite{Liao} fully characterized the capacity region for the general memoryless MAC, making it the only multi-user memoryless channel setting that is solved in its full generality.  Thus, the feasibility of describing two independent data sources without loss through a noisy channel with interference to a single decoder is solved.

Beyond the two cases mentioned, slight variations to the scenario result in a multitude of open problems in distributed source coding.  For example, the feasibility of describing two correlated data sources through a noisy MAC is not solved.  Furthermore, allowing the source coding to be done with loss raises even more uncertainty.  Berger and Tung \cite{Berger78} first considered the multiterminal source coding problem, where correlated sources are encoded separately with loss.  Even when no noisy channel is involved, the optimal rate region is not known, but ongoing progress continues \cite{Wagner--Tavildar--Viswanath2008} \cite{Krithivasan}.

The cascade multiterminal source coding setting is similar to multiterminal source coding considered by Berger and Tung in that two sources of information are encoded in a distributed fashion with loss.  The difference is that communication between the source encoders in this network replaces one of the direct channels to the decoder.  Thus, joint encoding is enabled to a degree, but the down side is that any message from Encoder 1 to the Decoder must now cross two links.

The general cascade multiterminal source coding problem includes many interesting variations.  The decoder may need to estimate both $X$ and $Y$, $X$ only, $Y$ only, or some other function of both, such as the sum of two jointly Gaussian random variables, considered in Section \ref{subsection gaussian}.  Vasudevan, Tian, and Diggavi \cite{Vasudevan} looked at a similar cascade communication system with a relay.  In their setting, the decoder has side information, and the relay has access to a physically degraded version of it.  Because of the degradation, the decoder knows everything it needs about the relay's side information, so the relay does not face the dilemma of mixing in some of the side information into its outgoing message.  In the cascade multiterminal source coding setting of this paper, the decoder does not have side information.  Thus, the relay is faced with coalescing the two pieces of information into a single message.  Other research involving similar network settings can be found in \cite{Gu06}, where Gu and Effros consider a more general network but with the restriction that the information $Y$ is a function of the information $X$, and \cite{Bakshi07}, where Bakshi et. al. identify the optimal rate region for lossless encoding of independent sources in a longer cascade (line) network.

In this paper we present inner and outer bounds on the general rate-distortion region for the cascade multiterminal source coding problem.  The inner bound addresses the challenge of compressing a sequence that is itself the result of a lossy compression.  Then we consider specific cases, such as encoding the sum of jointly Gaussian random variables, computing functions, and even coordinating actions.  The bounds are tight for computing functions and achieving some types of coordinated actions.

\section{Problem Specifics}
\label{section problem specifics}

The encoding of source symbols into messages is described in detail in the introduction and is depicted in Figure \ref{figure cascade multiterminal}.

\subsection{Objective}
\label{subsection objective}

The goal is for $X^n$, $Y^n$, and $Z^n$ to satisfy an average letter-by-letter distortion constraint $D$ with high probability.  A finite distortion function $d(x,y,z)$ specifies the penalty incurred for any triple $(x,y,z)$.  Therefore, the objective is to reliably produce a sequence $Z^n$ that satisfies
\begin{eqnarray}
\frac{1}{n} \sum_{i=1}^n d(X_i,Y_i,Z_i) & \leq & D.
\end{eqnarray}

Due to the flexibility in defining the distortion function $d$, the decoded sequence $Z^n$ can play a number of different roles.  If the goal is to estimate both sources $X$ and $Y$ with a distortion constraint, then $Z = (\hat{X}, \hat{Y})$ encompasses both estimates, and $d$ can be defined accordingly.  Alternatively, the decoder may only need to estimate $X$, in which case $Y$ acts as side information at a relay, and $Z = \hat{X}$.  In general, the decoder could produce estimates of any function of $X$ and $Y$.

\subsection{Rate-Distortion Region}
\label{subsection region}

The triple $(R_1,R_2,D)$ is an achievable rate-distortion triple for the distortion function $d$ and source distribution $p_0(x,y)$ if the following holds:
\begin{eqnarray*}
\mbox{For } \forall \epsilon > 0, \\
\exists n & \in & \{1,2,...\}, \\
\exists i & : & {\cal X}^n \rightarrow \{1,...,2^{nR_1}\}, \\
\exists j & : & {\cal Y}^n \times \{1,...,2^{nR_1}\} \rightarrow \{1,...,2^{nR_2}\}, \\
\exists z^n & : & \{1,...,2^{nR_2}\} \rightarrow {\cal Z}^n, \\
\mbox{such that} & & \\
& {\mathbb Pr} & \left(\frac{1}{n} \sum_{i=1}^n d(X_i,Y_i,Z_i) > D \right) \; < \; \epsilon, \\
\mbox{where} & & \\
Z^n & = & z^n(j(Y^n,i(X^n))).
\end{eqnarray*}

The rate-distortion region ${\cal R}$ for a particular source joint distribution $p_0(x,y)$ and distortion function $d$ is the closure of achievable rate-distortion triples, given as,
\begin{eqnarray}
{\cal R} & \triangleq & {\mathbb Cl} \{ \mbox{achievable } (R_1,R_2,D) \mbox{ triples} \}.
\end{eqnarray}

\section{General Inner Bound}
\label{section inner bound}

\begin{figure*}[b]
\vspace*{4pt}
\hrulefill
\begin{eqnarray}
\label{equation inner bound} {\cal R}_{in} & \triangleq & \left\{ (R_1,R_2,D) \; : \;
\begin{array}{l}
\exists p(x,y,z,u,v) = p_0(x,y)p(u,v|x)p(z|y,u,v) \mbox{ such that} \\
D > E(d(X,Y,Z)), \\
R_1 > I(X;U,V|Y), \\
R_2 > I(X;U) + I(Y,V;Z|U).
\end{array}
\right\}.
\end{eqnarray}
\end{figure*}

\begin{figure*}[b]
\begin{eqnarray}
\label{equation outer bound} {\cal R}_{out} & \triangleq & \left\{ (R_1,R_2,D) \; : \;
\begin{array}{l}
\exists p(x,y,z,u) = p_0(x,y)p(u|x)p(z|y,u) \mbox{ such that} \\
D \geq E(d(X,Y,Z)), \\
R_1 \geq I(X;U|Y), \\
R_2 \geq I(X,Y;Z).
\end{array}
\right\}.
\end{eqnarray}
\end{figure*}

The cascade multiterminal source coding network presents an interesting dilemma.  Encoder 2 has to summarize both the source sequence $Y^n$ and the message $I$ that describes $X^n$.  Intuition from related information theory problems, like Wyner-Ziv coding, suggests that for efficient communication the message $I$ from Encoder 1 to Encoder 2 will result in a phantom sequence of auxiliary random variables that are jointly typical with $X^n$ and $Y^n$ according to a selected joint distribution.  The second encoder could jointly compress the source sequence $Y^n$ along with the auxiliary sequence, treating it as if it was also a random source sequence.  But this is too crude.  A lot is known about the auxiliary sequence, such as the codebook it came from, allowing it to be summarized more easily than this approach would allow.  In some situations it proves more efficient to simply pass the description from Encoder 1 straight to the Decoder rather than to treat it as a random source and recompress at the second encoder.

While still allowing the message $I$ from Encoder 1 to be associated with a codebook of auxiliary sequences, we would like to take advantage of the sparsity of the codebook as we form a description at Encoder 2.  One way to accommodate this is to split the message from Encoder 1 into two parts.  One part is forwarded by Encoder 2, and the other part is decoded by Encoder 2 into a sequence of auxiliary variables and compressed with $Y^n$ as if it were a random source sequence.  The forwarded message keeps its sparse codebook in tact, while the decoded and recompressed message enjoys the efficiency that comes with being bundled with $Y$.  This results in an inner bound ${\cal R}_{in}$ for the rate-distortion region ${\cal R}$.  The definition of ${\cal R}_{in}$ is found in (\ref{equation inner bound}) at the bottom of this page.  The region ${\cal R}_{in}$ is already convex (for fixed $p_0(x,y)$ and $d$), so there is no need to convexify using time-sharing.

\begin{theorem}[Inner bound]
\label{theorem relay}
The rate-distortion region ${\cal R}$ for the cascade multiterminal source coding network of Figure \ref{figure cascade multiterminal} contains the region ${\cal R}_{in}$.  Every rate-distortion triple in ${\cal R}_{in}$ is achievable.  That is,
\begin{eqnarray}
{\cal R} & \supset & {\cal R}_{in}.
\end{eqnarray}
\end{theorem}

\begin{proof}
For lack of space, we give only a description of the encoding and decoding strategies involved in the proof and skip the probability of error analysis.  We use familiar techniques of randomized codebook construction, jointly typical encoding, and binning.

For any rate-distortion triple in ${\cal R}_{in}$ there is an associated joint distribution $p(x,y,z,u,v)$ that satisfies the inequalities in (\ref{equation inner bound}).  Construct three sets of codebooks, ${\cal C}_U$, ${\cal C}_{V,i}$, and ${\cal C}_{Z,i}$, for $i = 1,2,...,|{\cal C}_U|$, where
\begin{eqnarray*}
{\cal C}_U & = & \{u^n(i)\}_{i=1}^{m_1}, \\
{\cal C}_{V,i} & = & \{v^n(j,i)\}_{j=1}^{m_2}, \\
{\cal C}_{Z,i} & = & \{z^n(k,i)\}_{k=1}^{m_3}.
\end{eqnarray*}
Let $m_1 = 2^{n (I(X;U) + \epsilon)}$, $m_2 = 2^{n (I(X;V|U) + \epsilon)}$, and $m_3 = 2^{n (I(Y,V;Z|U) + \epsilon)}$.

Randomly generate the sequences $u^n(i) \in {\cal C}_U$ i.i.d. according to $p(u)$, independent for each $i$.  Then for each $i$ and $j$, independently generate the sequences $v^n(j,i) \in {\cal C}_{V,i}$ conditioned on $u^n(i) \in {\cal C}_U$ symbol-by-symbol according to $p(v|u)$.  Similarly, for each $i$ and $k$, independently generate the sequences $z^n(k,i) \in {\cal C}_{Z,i}$ conditioned on $u^n(i) \in {\cal C}_U$ symbol-by-symbol according to $p(z|u)$.

Finally, assign bin numbers.  For every sequence $u^n(i) \in {\cal C}_U$ assign a random bin $b_U(i) \in \{1,...,2^{n(I(X;U|Y) + 2\epsilon)}\}$.  Also, for each $i$ and each $v^n(j,i) \in {\cal C}_{V,i}$ assign a random bin $b_V(j,i) \in \{1,...,2^{n(I(X;V|Y,U) + 2\epsilon)}\}$.

Successful encoding and decoding is as follows.  Encoder 1 first finds a sequence $u^n(i) \in {\cal C}_U$ that is $\epsilon$-jointly typical with $X^n$ with respect to $p(x,u)$.  Then Encoder 1 finds a sequence $v^n(j,i) \in {\cal C}_{V,i}$ that is $\epsilon$-jointly typical with the pair $(X^n,u^n(i))$ with respect to $p(x,u,v)$.  Finally, Encoder 1 sends the bin numbers $b_U(i)$ and $b_V(j,i)$ to Encoder 2.

Encoder 2 considers all codewords in ${\cal C}_U$ with bin number $b_U(i)$ and finds that only $u^n(i)$ is $\epsilon$-jointly typical with $Y^n$ with respect to $p(y,u)$.  Then Encoder 2 considers all codewords in ${\cal C}_{V,i}$ with bin number $b_V(j,i)$ and finds that only $v^n(j,i)$ is $\epsilon$-jointly typical with the pair $(Y^n,u^n(i))$ with respect to $p(y,u,v)$.  Finally, Encoder 2 finds a sequence $z^n(k,i) \in {\cal C}_{Z,i}$ that is $\epsilon$-jointly typical with the triple $(Y^n,u^n(i),v^n(j,i))$ with respect to $p(y,u,v,z)$ and sends both $i$ and $k$ to the Decoder.

The decoder produces $Z^n = z^n(k,i)$.  Due to the Markov Lemma \cite{Berger78} and the structure of $p(x,y,z,u,v)$, the triple $(X^n,Y^n,Z^n)$ will be $\epsilon$-jointly typical with high probability.  Finally, $\epsilon$ can be chosen small enough to satisfy the rate and distortion inequalities.
\end{proof}

\section{General Outer Bound}
\label{section outer bound}

\begin{theorem}[Outer bound]
\label{theorem relay}
The rate-distortion region ${\cal R}$ for the cascade multiterminal source coding network of Figure \ref{figure cascade multiterminal} is contained in the region ${\cal R}_{out}$ defined in (\ref{equation outer bound}).  Rate-distortion triples outside of ${\cal R}_{out}$ are not achievable.  That is,
\begin{eqnarray}
{\cal R} & \subset & {\cal R}_{out}.
\end{eqnarray}
\end{theorem}

\begin{proof}
Identify the message $I$ from Encoder 1 along with the past and future variables in the sequence $Y^n$ as the auxiliary random variable $U$.
\end{proof}

\section{Special Cases}
\label{section special cases}

\subsection{Sum of Jointly Gaussian}
\label{subsection gaussian}
Suppose we wish to encode two jointly Gaussian data sources at Encoder 1 and Encoder 2 in order to produce an estimate of the sum at the decoder with small mean squared-error distortion.  Let $X$ and $Y$ be zero-mean jointly Gaussian random variables, where $X$ has variance $P_X$, $Y$ has variance $P_Y$, and their correlation coefficient is $\rho = \frac{E(XY)}{\sigma_X \sigma_Y}$.

\subsubsection{Inner bound}

We can explore the region ${\cal R}_{in}$ by optimizing over jointly Gaussian random variables $U$, $V$, and $Z$ to find achievable rate-distortion triples $(R_1,R_2,D)$.  This restricted search might not find all extremal rate-distortion points in ${\cal R}_{in}$; still it provides an inner bound on the rate-distortion region.  \footnote{To perform this optimization, first note that the marginal distribution $p(x,y,u)$ determines the quantities $I(X;U)$ and $I(X;U|Y)$, and $p(x,y,u)$ only has one significant free parameter due to Markovity.  All remaining quantities that define the region ${\cal R}_{in}$ are conditioned on $U$, including the final estimate at the decoder since $U$ is available to the decoder.  Therefore, after fixing $p(x,y,u)$ we can remove $U$ entirely from the optimization problem by exploiting the idiosyncracies of the jointly Gaussian distribution.  Namely, reduce the rates $R_1$ and $R_2$ appropriately and solve the problem without $U$ with $\bar{X}$ replacing $X$ and $\bar{Y}$ replacing $Y$, where $\bar{X}$ is the error in estimating $X$ with $U$, and $\bar{Y}$ is the error in estimating $Y$ with $U$.  This greatly reduces the dimensionality of the problem.}

The optimization of ${\cal R}_{in}$ with the restriction of only considering jointly Gaussian distributions $p(x,y,z,u,v)$ leads to two contrasting strategies depending on the variances $P_X$ and $P_Y$ of the sources and the rate $R_1$.  The two encoding strategies employed are to either forward the message from Encoder 1 to the Decoder, or to use the message to construct an estimate $\hat{X}^n$ at Encoder 2 and then compress the vector sum $\hat{X}^n + Y^n$ and send it to the Decoder, but not both.  In other words, either let $V=\emptyset$ ({\em forward} only) or let $U=\emptyset$ ({\em recompress} only).  The determining factor for deciding which method to use is a comparison of the rate $R_1$ with the quantity $\frac{1}{2} \log_2 \frac{P_X}{P_Y}$.

{\em Case 1:  (Recompress)}
\begin{eqnarray*}
R_1 & \geq & \frac{1}{2} \log_2 \frac{P_X}{P_Y}.
\end{eqnarray*}
If the rate $R_1$ is large enough, then the optimal encoding method is to recompress at Encoder 2.  This will allow for a more efficient encoding of the sum in the second message $J$ rather than encoding two components of the estimate separately.

The distortion in this case is
\begin{eqnarray}
D & = & (1 - \rho^2) \left( 1 - 2^{-2 R_2} \right) 2^{-2 R_1} P_X \nonumber \\
& & + \;\; 2^{-2 R_2} P_{X+Y}, \label{equation gaussian recompress}
\end{eqnarray}
where $P_{X+Y}$ is the variance of the sum $X+Y$.

{\em Case 2:  (Forward)}
\begin{eqnarray*}
R_1 & < & \frac{1}{2} \log_2 \frac{P_X}{P_Y}.
\end{eqnarray*}
If the variance of $X$ is larger than $Y$ and the rate $R_1$ is small, then the optimal encoding method is to forward the message $I$ from Encoder 1 to the Decoder without changing it.  By rearranging the inequality, we see that $2^{-2 R_1} P_X > P_Y$.  From rate-distortion theory we know that $2^{-2 R_1} P_X$ is the mean squared-error that results from compressing $X$ at rate $R_1$.  The fact that the variance of the error introduced by the compression at Encoder 1 is larger than the variance of $Y$ subtly indicates that the description of $X$ was more efficiently compressed by Encoder 1 than it would be if mixed with $Y$ and recompressed.

The estimate of $X$ from Encoder 1, represented by $U$, which is forwarded by Encoder 2, might be limited by either $R_1$ or $R_2$.  In the case that $R_2$ is completely saturated with the description of $U$ at rate $I(X;U)$, there is no use trying to use any excess rate $R_1 - I(X;U|Y)$ from Encoder 1 to Encoder 2 because it will have no way of reaching the decoder. On the other hand, in the case that $R_1$ is the limiting factor for the description of $U$ at rate $I(X;U|Y)$, then the excess rate $R_2 - I(X;U)$ can be used to describe $Y$ to the decoder.  We state the distortion separately for each of these cases.

If $R_2 \leq \frac{1}{2} \log_2 \left( \frac{2^{2 R_1} - \rho^2}{1 - \rho^2} \right)$ then,
\begin{eqnarray*}
D & = & 2^{-2 R_2} \left( P_{X+Y} + (1 - \rho^2) \left( 2^{2 R_2} - 1 \right) P_Y \right).
\end{eqnarray*}
If $R_2 > \frac{1}{2} \log_2 \left( \frac{2^{2 R_1} - \rho^2}{1 - \rho^2} \right)$ then,
\begin{eqnarray*}
D & = & \left( (1-\rho^2) 2^{-2 R_1} - \left( 1 - \rho^2 2^{-2 R_1} \right) 2^{-2 R_2} \right) P_X \\
& & + \;\; 2^{-2 R_2} \left( P_{X+Y} + \left( 2^{2 R_1} - 1 \right) P_Y \right).
\end{eqnarray*}
Again, $P_{X+Y}$ is the variance of the sum $X+Y$.

\subsubsection{Outer bound}

The outer bound ${\cal R}_{out}$ is optimized with Gaussian auxiliary random variables.  However, for simplicity, we optimize an even looser bound by minimizing $R_1$ and $R_2$ separately (cut-set bound) for a given distortion constraint.  The result is the following lower bound on distortion.
\begin{eqnarray}
D & \geq & \max \{ 2^{-2R_1} (1-\rho^2) P_X,  \; \; 2^{-2R_2} P_{X+Y} \}. \label{equation gaussian outer bound}
\end{eqnarray}

\subsubsection{Sum-Rate}

Consider the sum-rate $R_1 + R_2$ required to achieve a given distortion level $D$.  We can compare the sum-rate-distortion function $R(D)$ for the inner and outer bounds.

Let $P_X \leq P_Y$.  This puts us in the recompress regime of the inner bound.  By optimizing (\ref{equation gaussian recompress}) subject to $R_1 + R_2 = R$, we find that the optimal values $R_1^*$ and $R_2^*$ satisfy
\begin{eqnarray*}
R_2^* - R_1^* & = & \frac{1}{2} \log_2 \left( \frac{P_{X+Y}}{(1-\rho^2) P_X} \right),
\end{eqnarray*}
as long as $R$ is greater than the right-hand side.  Notice that $R_2$ is more useful than $R_1$, as we might expect.  From this we find a piece-wise upper bound on the sum-rate-distortion function.  Similarly we find a piece-wise lower bound based on (\ref{equation gaussian outer bound}).

{\em Sum-rate upper bound.}  Low distortion region:
\begin{eqnarray*}
D & \leq & (1-\rho^2) P_X \left( 2 - \frac{(1-\rho^2) P_X}{P_{X+Y}} \right),
\end{eqnarray*}
then
\begin{eqnarray*}
R(D) & \leq & \frac{1}{2} \log_2 \left( \frac{P_{X+Y}}{D} \right) + \frac{1}{2} \log_2 \left( \frac{(1-\rho^2) P_X}{D} \right) \\
& & + \;\; \log_2 \left( 1 + \sqrt{1 - \frac{D}{P_{X+Y}}} \right).
\end{eqnarray*}
High distortion region: (up to $D \leq P_{X+Y}$)
\begin{eqnarray*}
R(D) & \leq & \frac{1}{2} \log_2 \left( \frac{P_{X+Y} - (1-\rho^2) P_X}{D - (1-\rho^2) P_X} \right).
\end{eqnarray*}

{\em Sum-rate lower bound.}  Low distortion region:
\begin{eqnarray*}
D & \leq & (1-\rho^2) P_X,
\end{eqnarray*}
then
\begin{eqnarray*}
R(D) & \geq & \frac{1}{2} \log_2 \left( \frac{P_{X+Y}}{D} \right) + \frac{1}{2} \log_2 \left( \frac{(1-\rho^2) P_X}{D} \right).
\end{eqnarray*}
High distortion region: (up to $D \leq P_{X+Y}$)
\begin{eqnarray*}
R(D) & \geq & \frac{1}{2} \log_2 \left( \frac{P_{X+Y}}{D} \right).
\end{eqnarray*}

\begin{lemma}
\label{lemma 1 bit gap}
The gap between the upper and lower bounds on the optimal sum-rate (derived from ${\cal R}_{in}$ and ${\cal R}_{out}$) needed to encode the sum of jointly Gaussian sources in the cascade multiterminal network with a squared-error distortion constraint $D$ is no more than 1 bit, shrinking as $D$ increases, for any jointly Gaussian sources satisfying $P_X \leq P_Y$.
\end{lemma}

\subsection{Computing a Function}
\label{subsection function}

Instead of estimating a function of $X$ and $Y$, we might want to compute a function exactly.  Here we show that the bounds ${\cal R}_{in}$ and ${\cal R}_{out}$ are tight for this lossless case.\footnote{The optimal rate region for computing functions of data in the standard multiterminal source coding network is currently an open problem \cite{HanKobayashi87}.}  To do so, we consider an arbitrary point $(R_1,R_2,D) \in {\cal R}_{out}$ and its associated distribution $\hat{p}(x,y,z,u)$.  For the inner bound ${\cal R}_{in}$ we use the same distribution $\hat{p}$; however, let $U=\emptyset$ and $V$ take the role of $U$ from the outer bound.  Notice that the Markovity constraints are satisfied.  Now consider,
\begin{eqnarray*}
I(Y,V;Z) & = & H(Z) - H(Z|Y,V) \\
& = & H(Z) - H(Z|Y,V,X) \\
& = & H(Z) \\
& = & I(X,Y;Z),
\end{eqnarray*}
due to the Markovity constraint $X - (Y,V) - Z$ and the fact that $Z$ is a function of $X$ and $Y$.  Therefore, for this distribution $\hat{p}$, all of the inequalities in ${\cal R}_{in}$ are satisfied for the point $(R_1,R_2,D)$.

The outer bound ${\cal R}_{out}$ makes it clear that optimal encoding is achieved by using Wyner-Ziv encoding from Encoder 1 to compute the value of the function $Z$ at Encoder 2.  This optimization is carefully investigated in \cite{Orlitsky01} and equated to a graph entropy problem.  Then Encoder 2 compresses $Z$ to the entropy limit.

\subsection{Markov Coordination}
\label{subsection markov}

It is possible to talk about achieving a joint distribution of coordinated actions $p(x,y,z) = p_0(x,y)p(z|x,y)$ without referring to a distortion function, as in \cite{Cover07Permuter}.  Under some conditions of the joint distribution, the bounds ${\cal R}_{in}$ and ${\cal R}_{out}$ are tight.  One obvious condition is when $X$, $Y$, and $Z$ form the Markov chain $X-Y-Z$.  In this case, there is no need to send a message $I$ from Encoder 1, and the only requirement for achievability is that $R_2 \geq I(Y;Z)$.

Another class of joint distributions $p_0(x,y)p(z|x,y)$ for which the rate bounds are provably tight is all distributions forming the Markov chain $Y-X-Z$.  This encompasses the case where $Y$ is a function of $X$, as in \cite{Gu06}.  To prove that the bounds are tight, choose $U = Z$ and $V = \emptyset$ for ${\cal R}_{in}$.  We find that rate pairs satisfying
\begin{eqnarray}
R_1 & \geq & I(X;Z|Y), \\
R_2 & \geq & I(X;Z),
\end{eqnarray}
are achievable.  And all rate pairs in ${\cal R}_{out}$ satisfy these inequalities.

\section{Acknowledgment}
The authors would like to recognize Haim Permuter's various contributions and curiosity toward this work.  This work is supported by the National Science Foundation through grants CCF-0515303 and CCF-0635318.

\bibliographystyle{unsrt}
\bibliography{isit}

\end{document}